\documentclass[3p,twocolumn]{elsarticle}

\usepackage{grproc}
\usepackage{lineno,hyperref}
\modulolinenumbers[5]

\journal{Information Processing Letters}

\begin{document}

\begin{frontmatter}

\title{Unifying Theories of Time with Generalised Reactive Processes}

\author[cor1]{Simon Foster}
\ead{simon.foster@york.ac.uk}
\author{Ana Cavalcanti}
\author{Jim Woodcock}
\author{Frank Zeyda}

\cortext[cor1]{Corresponding author}

\address{Department of Computer Science, University of York, York, YO10 5DD, United Kingdom}

\begin{abstract}
  Hoare and He's theory of reactive processes provides a unifying foundation for the formal semantics of concurrent and
  reactive languages. Though highly applicable, their theory is limited to models that can express event histories as
  discrete sequences. In this paper, we show how their theory can be generalised by using an abstract trace algebra. We
  show how the algebra, notably, allows us to also consider continuous-time traces and thereby facilitate models of
  hybrid systems. We then use this algebra to reconstruct the theory of reactive processes in our generic setting, and
  prove characteristic laws for sequential and parallel processes, all of which have been mechanically verified in the
  Isabelle/HOL proof assistant.
\end{abstract}


\end{frontmatter}


\section{Introduction}
\label{sec:intro}
The theory of reactive processes provides a generic foundation for denotational semantics of concurrent languages. It
was created as part of the Unifying Theories of Programming (UTP)~\cite{Hoare&98,Cavalcanti&06} framework, which models
computation using predicate calculus. The theory of reactive processes unifies formalisms such as CSP~\cite{Hoare85},
ACP~\cite{Bergstra84}, and CCS~\cite{Milner89}. This is made possible by its support of a large set of algebraic
theorems that universally hold for families of reactive languages. The theory has been extended and applied to several
languages including, stateful~\cite{Oliveira&09} and real-time languages, with both discrete~\cite{Wei2013} and
continuous time~\cite{Foster16b,He94}.

Technically, the theory's main feature is its trace model, which provides a way for a process to record an interaction
history, using an observational variable $tr : \seq~\emph{Event}$. In the original presentation, a trace is a discrete
event sequence, which is standard for languages like CSP. The alphabet can be enriched by adding further observational
variables; for example, $ref : \power~\emph{Event}$ to model refusals~\cite{Hoare&98}.

Though sequence-based traces are ubiquitous for modelling concurrent systems, other models exist. In particular, the
sequence-based model is insufficient to represent continuous evolution of variables as present in hybrid systems. A
typical notion of history for continuous-time systems are real-valued trajectories $\real_{\ge 0} \to \Sigma$ over
continuous state $\Sigma$.



Although the sequence and trajectory models appear substantially different, there are many similarities. For example, in
both cases one can subdivide the history into disjoint parts that have been contributed by different parts of the
program, and describe when a trace is a prefix of another. By characterising traces abstractly, and thus unifying
these different models, we provide a generalised theory of reactive processes whose properties, operators, and laws can
be transplanted into an even wider spectrum of languages. We thus enable unification of untimed, discrete-time, and
continuous-time languages. The focus of our theory is on traces of finite length, but the semantic framework is
extensible.

We first introduce UTP and its applications (\S{\ref{sec:bg}}). We then show how traces can be characterised
algebraically by a form of cancellative monoid (\S{\ref{sec:traces}}), and that this algebra encompasses both sequences
and piecewise continuous functions (\S{\ref{sec:models}}). We apply this algebra to generalise the theory of reactive
processes, and show that its key algebraic laws are retained in our generalisation, including those for sequential and
parallel composition (\S{\ref{sec:grp}}).

Our work is mechanised in our theorem prover,
Isabelle/UTP\footnote{\url{https://github.com/isabelle-utp/utp-main}}~\cite{Foster16a}, a semantic embedding of UTP in
Isabelle/HOL. We sometimes give proofs, but these merely illustrate the intuition, with the mechanisation being
definitive. To the best of our knowledge, this is the most comprehensive mechanised account of reactive processes.

\section{Background}
\label{sec:bg}
UTP is founded on the idea of encoding program behaviour as relational predicates whose variables correspond to
observable quantities. Unprimed variables ($x$) refer to observations at the start, and primed variables ($x'$) to
observations at a later point of the computation. The operators of a programming language are thus encoded in
predicate calculus, which facilitates verification through theorem proving. For example, we can specify sequential
programming operators as relations:
\begin{align*}
  x := v       & ~~\defs~~ x' = v \land y_1' = y_1 \land \cdots \land y_n' = y_n \\
  P \relsemi Q & ~~\defs~~ \exists x_0 @ P[x_0/x'] \land Q[x_0/x] \\
  \conditional{P}{b}{Q} & ~~\defs~~ (b \wedge P) \vee (\neg b \wedge Q)
\end{align*}
\noindent Assignment $x := v$ states that $x'$ takes the value $v$ and all other variables are unchanged. We define the
degenerate form $\II \defs x := x$, which identifies all variables. Sequential composition $P \relsemi Q$ states that
there exists an intermediate state $x_0$ on which $P$ and $Q$ agree. If-then-else conditional $\conditional{P}{b}{Q}$
states that if $b$ is true, $P$ executes, otherwise $Q$.

UTP variables can either encode program data, or behavioural information, in which case they are called
\emph{observational} variables. For example, we may have $ti, ti' : \mathbb{R}_{\ge 0}$ to record the time before
and after execution. These exist to enrich the semantic model and are constrained by \emph{healthiness
  conditions} that restrict permissible behaviours. For example, we can impose $ti \le ti'$ to forbid reverse time
travel.

Healthiness conditions are expressed as functions on predicates, such as $\ckey{HT}(P) \defs P \land ti \le ti'$,
the application of which coerces predicates to healthy behaviours. When such functions are idempotent and monotonic,
with respect to the refinement order $\refinedby$, we can show, with the aid of the Knaster-Tarski theorem, that their
image forms a complete lattice, which allows us to reason about recursion.

Healthiness conditions are often built from compositions:
$\healthy{H} \defs \healthy{H}_1 \circ \healthy{H}_2 \circ \cdots \circ \healthy{H}_n$.  In this case, idempotence and
monotonicity of $\healthy{H}$ can be shown by proving that each $\healthy{H}_i$ is monotonic and idempotent, and each
$\healthy{H}_i$ and $\healthy{H}_j$ commute. A set of healthy fixed-points,
$\carrier{\healthy{H}} \defs \{P | \healthy{H}(P) = P\}$, is called a \emph{UTP theory}. Theories isolate the
aspects of a programming language, such as concurrency, object orientation, and real-time programming. Theories can also
be combined by composing their healthiness conditions to enable construction of sophisticated heterogeneous and
integrated languages.

Our focus is the theory of reactive processes, with healthiness condition $\healthy{R}$, which we formalise in
Section~\ref{sec:grp}. Reactive programs, in addition to initial and final states, also have intermediate states, during
which the process waits for interaction with its environment. $\healthy{R}$ specifies that processes yield well-formed
traces, and that, when a process is in an intermediate state, any successor must wait for it to terminate before
interacting. This theory uses observational variable $wait$ to differentiate intermediate from final states, and $tr$
to record the trace.

UTP theories based on reactive processes have been applied to give formal semantics to a variety of
languages~\cite{Hoare&98,Smith2005,Zhu2015}, notably the \Circus formal modelling language family~\cite{Oliveira&09},
which combines stateful modelling, concurrency, and discrete time~\cite{Wei2013,Woodcock14}. A similar theory has been
used for a hybrid variant of CSP~\cite{He94}, with a modified notion of trace. Though sharing some similarities, these
various versions of reactive processes are largely disjoint theories with distinct healthiness conditions. Our
contribution is to unify them all under the umbrella of \emph{generalised reactive processes}.

\section{Trace Algebra}
\label{sec:traces}
In this section, we describe the trace algebra that underpins our generalised theory of reactive processes. We define
traces as an abstract set $\tset$ equipped with two operators: trace concatenation $\tcat : \tset \to \tset \to \tset$,
and the empty trace $\tempty : \tset$, which obey the following axioms.

\begin{definition}[Trace Algebra] \label{def:tralg} A trace algebra $(\tset, \tcat, \tempty)$ is a cancellative monoid satisfying the
  following axioms:
\begin{align*}
  x \tcat (y \tcat z) &= (x \tcat y) \tcat z \tag{TA1} \label{law:tassoc} \\
  \tempty \tcat x = x \tcat \tempty &= x \tag{TA2} \label{law:tunit} \\
  x \tcat y = x \tcat z ~~\implies~~  y &= z \tag{TA3} \label{law:tcancl1} \\
  x \tcat z = y \tcat z ~~\implies~~  x &= y \tag{TA4} \label{law:tcancl2} \\
  x \tcat y = \tempty  ~~\implies~~  x &= \tempty \tag{TA5} \label{law:tnai}
\end{align*}
\end{definition}

\noindent As expected, $\tcat$ is associative and has left and right units.  Axioms~\ref{law:tcancl1} and
\ref{law:tcancl2} show that $\tcat$ is injective in both arguments. As an aside, \ref{law:tcancl1} and \ref{law:tcancl2}
hold only in models without infinitely long traces, as such a trace $x$ would usually annihilate $y$ in $x \tcat
y$. Axiom~\ref{law:tnai} states that there are no ``negative traces'', and so if $x$ and $y$ concatenate to $\tempty$ then
$x$ is $\tempty$. We can also prove the dual law: $x \tcat y = \tempty ~~\implies~~ y = \tempty$. From this algebraic
basis, we derive a prefix relation and subtraction operator.

\begin{definition}[Trace Prefix and Subtraction] \label{def:trops}
\begin{align*}
  x \le y &~~\iff~~ \exists z @ y = x \tcat z \\[2ex]
  y \tminus x &~~\defs~~ \begin{cases} \iota z @ y = x \tcat z & \text{if}~x \le y \\ \tempty & \text{otherwise} \end{cases}
\end{align*}
\end{definition}

\noindent Trace prefix, $x \le y$, requires that there exists $z$ that extends $x$ to yield $y$. Trace subtraction $y
\tminus x$ obtains that trace $z$ when $x \le y$, using the definite description operator (Russell's $\iota$), and
otherwise yields the empty trace. This is slightly different from the standard UTP operator, which is defined only when
$x \le y$. We can prove the following laws about trace prefix.

\begin{theorem}[Trace Prefix Laws] For $x, y, z : \tset$,
  \begin{align*}
    (\tset, \le) & \text{ is a partial order} \tag{TP1} \label{law:TP1} \\
    \tempty &\le x \tag{TP2} \label{law:tp0le} \\
    x &\le x \tcat y \tag{TP3} \label{law:TP3} \\
    x \tcat y \le x \tcat z &\iff y \le z \tag{TP4} \label{law:TP4}
  \end{align*}
\end{theorem}

\noindent \ref{law:tp0le} tells us that $\tempty$ is the smallest trace, \ref{law:TP3} that concatenation builds larger
traces, and \ref{law:TP4} that concatenation is monotonic in its right argument. We also have the following trace
subtraction laws.


\begin{theorem}[Trace Subtraction Laws] \label{thm:trsubtract}
\begin{align*}
  x \tminus \tempty &~=~ x \tag{TS1} \label{law:tsright0} \\
  \tempty \tminus x &~=~ \tempty \tag{TS2} \\
  x \tminus x &~=~ \tempty \tag{TS3} \label{law:TS3} \\
  (x \tcat y) \tminus x &~=~ y \tag{TS4} \label{law:tscat} \\
  (x \tminus y) \tminus z &~=~ x \tminus (y \tcat z) \tag{TS5} \label{law:TS5} \\
  (x \tcat y) \tminus (x \tcat z) &~=~ y \tminus z \tag{TS6} \label{law:TS6} \\
  y \le x \land x \tminus y = \tempty &~\iff~ x = y \tag{TS7} \label{law:TS7} \\
  x \le y &~\implies~  x \tcat (y \tminus x) = y \label{law:TS8} \tag{TS8}
\end{align*}
\end{theorem}

\noindent Laws \ref{law:tsright0}-\ref{law:TS3} relate trace subtraction and the empty trace. \ref{law:tscat} shows that
subtraction inverts concatenation. \ref{law:TS5} shows that subtracting two traces is equivalent to subtracting their
concatenation. \ref{law:TS6} shows that subtraction can be used to remove a common prefix. \ref{law:TS7} shows that two
traces are equal if, and only if, the first is a prefix of the second and they subtract to $\tempty$. \ref{law:TS8}
shows that a trace can be split into its prefix and suffix.

In the next section, we show that standard notions of traces are models. Afterwards, in Section~\ref{sec:grp} we use the
algebra to create the generalised theory of reactive processes.



\section{Trace Models}
\label{sec:models}
In this section we describe three trace models: positive reals, finite sequences, and timed traces. Other models are
possible; for example, we can further extend timed traces to ``super-dense time''~\cite{Lee2014} to encompass multiple
distinguished discrete state updates at a time instant. We leave study of other models as future work.

Positive real numbers $\mathbb{R}_{\ge 0}$ form one of the simplest models of the trace algebra.

\begin{theorem} $(\mathbb{R}_{\ge 0}, +, 0)$ is a trace algebra. \end{theorem}

\begin{proof}
  $+$ is clearly associative, cancellative, and has $0$ as its left and right unit. Moreover, since $+$ is commutative
  and $\mathbb{R}_{\ge 0}$ contains no negative numbers then $+$ has no additive inverse.
\end{proof}

Positive reals can be used to express timed programs with a clock variable $time : \mathbb{R}_{\ge
  0}$~\cite{Hayes2010}. Finite sequences, unsurprisingly, also form a trace algebra, when we set 
$\tcat$ to sequence concatenation ($\cat$), and $\tempty$ to the empty sequence ($\langle\rangle$).

\begin{theorem} $(\seq\,\textit{Event}, \cat, \langle\rangle)$ is a trace algebra. \end{theorem}

\noindent Though simple, we note that the sequence-based trace model has been shown to be sufficient to characterise
both untimed~\cite{Oliveira&09} and discrete time modelling languages~\cite{Woodcock14}.

A more complex model is that of piecewise continuous functions, for which we adopt and refine a model called \emph{timed
  traces} ($\ttrtype$)~\cite{Hayes2006}. A timed trace is a partial function of type $\real_{\ge 0} \pfun \Sigma$, for
continuous state type $\Sigma$, which represents the system's continuous evolution with respect to time.

\begin{figure}[t]
  \begin{center}
    \includegraphics[width=7.5cm]{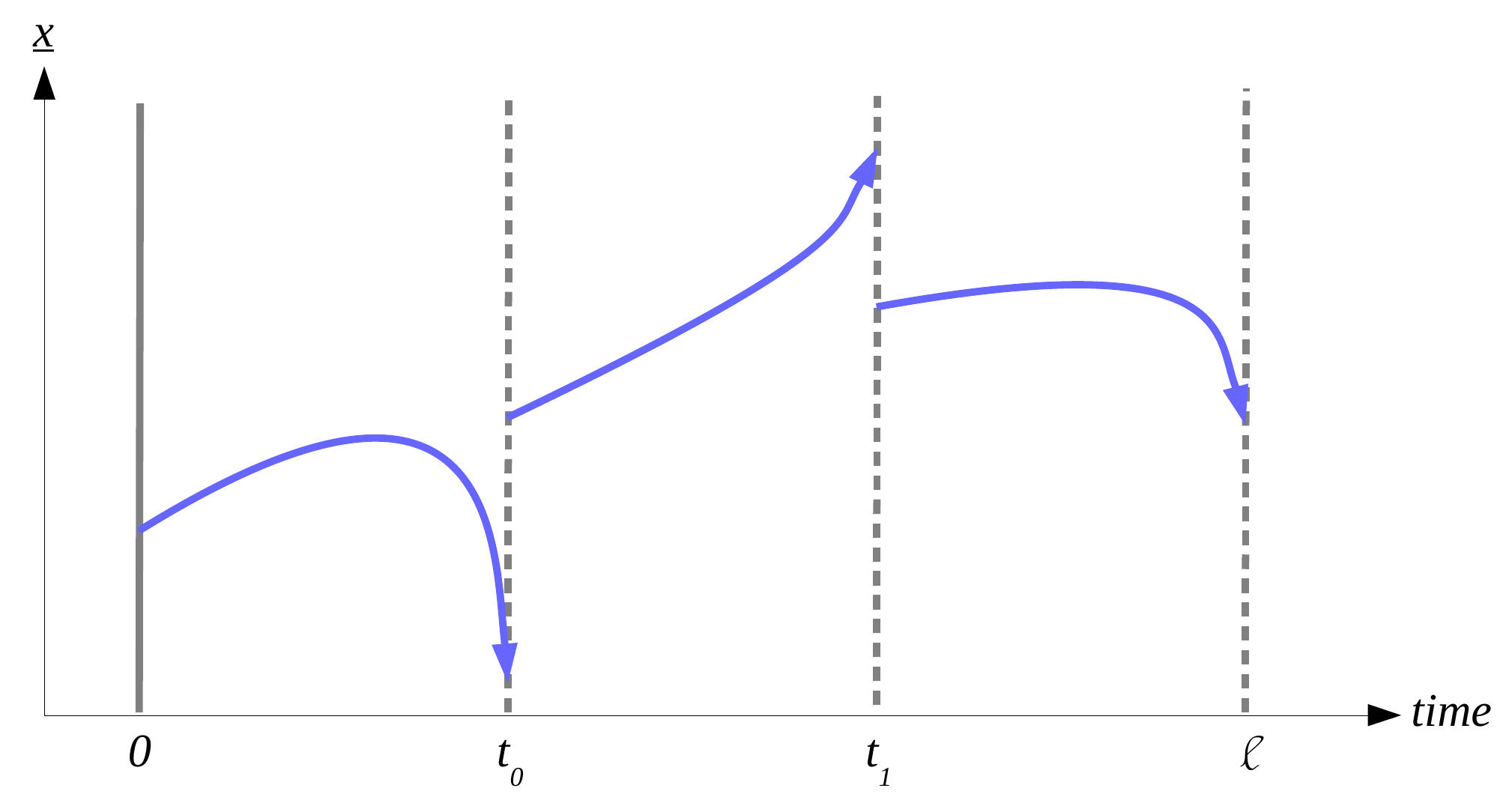}
  \end{center}

  \vspace{-3ex}
  \caption{Piecewise continuous timed traces}
  \vspace{-1ex}
  \label{fig:ttrace}
\end{figure}

In our model we also require that timed traces be piecewise continuous, to allow both continuous and discrete
information. A timed trace is split into a finite sequence of continuous segments, as shown in Figure~\ref{fig:ttrace}.
Each segment accounts for a particular evolution of the state interspersed with discontinuous discrete events. This
necessitates that we can describe limits and continuity, and consequently we require that $\Sigma$ be a topological
space, such as $\mathbb{R}^n$, though it can also contain discrete topological information, like events. Continuous
variables are projections such as $x : \Sigma \to \mathbb{R}$. We give the formal model below.
\begin{definition}[Timed Traces]
$$\begin{array}{ll}
    \ttrtype \defs & \!\! 
    \left\{\begin{array}{ll}
    f : \mathbb{R}_{\ge 0} \pfun \Sigma \\
    |~ \exists t : \real_{\ge 0} @ \dom(f) = [0, t) \\
    \qquad \land t > 0 \implies \exists I : \mathbb{R}_{\textnormal{oseq}} 
    \\  
    ~~ @ \left( \begin{array}{l}
      \ran(I) \subseteq [0,t] \\ 
      \land \{0,t\} \subseteq \ran(I) \\
      \land \left( \begin{array}{l} \forall n < \# I - 1 @ \\ ~ f \conton [I_n, I_{n+1})
        \end{array} \! \right)
      \end{array} \! \right)
  \end{array}
  \right\}
\end{array}$$
$\begin{array}{l}
  \orseq \defs \{x : \textnormal{seq}\,\mathbb{R} ~|~ \forall n < \#{x} - 1 @ x_n < x_{n + 1}\} \\[1ex]
  f \conton [m, n) \defs \forall t \in [m, n) @ \underset{x \to t^{+}}{\lim}\, f(x) = f(t)
\end{array}$
\end{definition}
\noindent A timed trace is a partial function $f$ with domain $[0, t)$, for end point $t \ge 0$. When the trace is
non-empty $(t > 0)$, there exists an ordered sequence of instants $I$ giving the bounds of each segment. $\orseq$ is the
subset of finite real sequences such that for every index $n$ in the sequence less than its length $\#{x}$, $x_n <
x_{n+1}$. $I$ must naturally contain at least $0$ and $t$, and only values between these two extremes. The timed trace
$f$ is required to be continuous on each interval $[I_n, I_{n+1})$. The operator $f \conton A$ denotes that $f$ is
continuous on the range given by $A$. We now introduce the core timed trace operators, which take inspiration from
H\"{o}fner's algebraic trajectories~\cite{Hofner2009}.
\begin{definition}[Timed-trace Operators]
$$\begin{array}{rcl}
  \ttend(f) &~~\defs~~& \mathrm{min}(\mathbb{R}_{\ge 0} \setminus \dom(f)) \\[1ex]
  \tempty &~~\defs~~& \emptyset \\[1ex]
  f \tcat g &~~\defs~~& f \cup (g \rshift \ttend(f))
\end{array}$$
\end{definition}
\noindent Auxiliary function $f \rshift n$ shifts the indices of a partial function $f : \real_{\ge 0} \pfun A$ to the
right by $n : \real_{\ge 0}$, and has definition $\lambda x @ f(x - n)$. The operator $\ttend(f)$ gives the end time of a trace
$f : \ttrtype$ by taking the infimum of the real numbers excluding the domain of $f$. The empty trace $\tempty$ is the
empty function. Finally, $f \tcat g$ shifts the domain of $g$ to start at the end of $f$, and takes the union. We
establish laws governing these trace operators.
\begin{theorem}[Timed-trace Laws]
  \begin{align*}
  (f \rshift m) \rshift n &= f \rshift (m + n) \tag{T1} \label{law:T1} \\
  (f \cup g) \rshift n &=   (f \rshift n) \cup (g \rshift n) \tag{T2} \label{law:T2} \\
  \ttend(\tempty) &= 0 \tag{T3} \label{law:T3} \\
  \ttend(x \tcat y) &= \ttend(x) + \ttend(y) \tag{T4} \label{law:T4}
\end{align*}
\end{theorem}

\noindent \ref{law:T1} shows that shifting a function twice equates to a single shift on their summation. \ref{law:T2}
shows that shift distributes through function union. \ref{law:T3} shows that the length of the empty trace is $0$, and
\ref{law:T4} shows that the length of a trace is the sum of its parts. $\ttrtype$ is closed under trace concatenation.

\begin{theorem}[Trace Concatenation Closure] If there exists $m, n : \real_{\ge 0}$, such that $\dom(tt_1) = [0,m)$, and
  $\dom(tt_2) = [0,n)$, then $tt_1,tt_2 \in \ttrtype$ if, and only if, $tt_1 \tcat tt_2 \in \ttrtype$.
\end{theorem}

\noindent This theorem tells us that decomposition of a timed trace always yields timed traces, provided both $f$ and
$g$ have a contiguous domain. Finally, trace concatenation satisfies our trace algebra.

\begin{theorem} $(\ttrtype, \tcat, \tempty)$ forms a trace algebra
\end{theorem}

\begin{proof} For illustration, we show the derivation for associativity. The other proofs are simpler.
  \begin{align*}
     x & \tcat (y \tcat z) \\
  &= x \cup ((y \cup (z \rshift \ttend(y)) \rshift \ttend(x)) \\
  &= x \cup ((y \rshift \ttend(x)) \cup (z \rshift \ttend(y) \rshift \ttend(x))) \\
  &= (x \cup (y \rshift \ttend(x))) \cup (z \rshift (\ttend(x) + \ttend(y))) \\
  &= (x \tcat y) \cup (z \rshift (\ttend(x) + \ttend(y))) \\
  &= (x \tcat y) \cup (z \rshift (\ttend(x \tcat y))) \\
  &= (x \tcat y) \tcat z
  \end{align*} 

  \vspace{-4.5ex}

\end{proof}

\noindent This model provides the basis for hybrid computation. We introduce the theory in the next section.

\section{Generalised Reactive Processes}
\label{sec:grp}
Here, we use our trace algebra to provide a generalised theory of reactive processes. We prove the key laws of reactive
processes, thus demonstrating the conservative nature of our theory. Many of the properties here have been previously
proved~\cite{Cavalcanti&06}, but we restate and prove many of them due to our weakening of the trace model and some
small differences. Another novelty is that all these theorems have been mechanised in our Isabelle/UTP
repository. Following~\cite{Hoare&98,Cavalcanti&06} we define the theory in terms of two pairs of observational
variables:

\begin{itemize}
\item $wait, wait' : \Bool$ -- describe when the previous or current process, respectively, is in an intermediate state;
\item $tr, tr' : \tset$ -- the trace that occurred prior to and after execution of the current process in terms of a
  trace algebra $(\tset, \tcat, \tempty)$.
\end{itemize}

\noindent Our theory does not contain refusal variables $ref, ref'$, as these are not always necessary to describe
reactive processes~\cite{Woodcock14}. We describe three healthiness conditions namely \healthy{R1}, $\healthy{R2}_c$,
and \healthy{R3}. $\healthy{R1}$ and $\healthy{R3}$ are already presented in \cite{Hoare&98}; for their $\healthy{R2}$
we have a different formulation, which we call $\healthy{R2}_c$.
\begin{definition}[Reactive Healthiness Conditions]
  \begin{align*}
    \healthy{R1}(P) & ~\defs~ P \land tr \le tr' \\
    \healthy{R2}_c(P) & ~\defs~ \conditional{P[\tempty, tr' \tminus tr/tr, tr']}{tr \le tr'}{P} \\
    \healthy{R3}(P) & ~\defs~ \conditional{\II}{wait}{P} \\
    \healthy{R} & ~\defs~ \healthy{R3} \circ \healthy{R2}_c \circ \healthy{R1}
  \end{align*}
\end{definition}
\noindent $\healthy{R1}$ states that $tr$ is monotonically increasing; processes are not permitted to undo past
events. $\healthy{R2}_c$ states that a process must be history independent: the only part of the trace it may constrain
is $tr' \tminus tr$, that is, the portion since the previous observation $tr$. Specifically, if the history is deleted,
by substituting $\tempty$ for $tr$ and $tr' - tr$ for $tr'$, then the behaviour of the process is unchanged. Our
formulation of $\healthy{R2}_c$ deletes the history only when $tr \le tr'$, which ensures that $\healthy{R2}_c$ does not
depend on $\healthy{R1}$, and thus commutes with it. Finally, $\healthy{R3}$ states that if a prior process is
intermediate ($wait'$) then the current process must identify all variables.

We compose the three to yield $\healthy{R}$, the overall healthiness condition of reactive processes. An example
$\healthy{R}$ healthy predicate is $$\healthy{R3}(tr' = tr \tcat \langle a \rangle \land v' = v)$$ which extends the
trace with a single event $a$ and leaves program variable $v$ unchanged. We show that $\healthy{R}$ is idempotent and
monotonic.
\begin{theorem}[$\healthy{R}$ idempotence and monotonicity] \label{thm:reah-im}
$$\healthy{R} = \healthy{R} \circ \healthy{R} \text{~~and~~} P \refinedby Q \implies \healthy{R}(P) \refinedby \healthy{R}(Q)$$
\end{theorem}
\noindent A corollary of Theorem~\ref{thm:reah-im} is that reactive processes form a complete lattice.

\begin{theorem}
  Reactive processes form a complete lattice ordered by $\refinedby$, with infimum $\thinf{R} \, A$ and supremum
  $\thsup{R} \, A$, for $A \subseteq \carrier{\healthy{R}}$.
\end{theorem}


\noindent This, in particular, provides us with specification and reasoning facilities about recursive reactive processes
using the fixed-point operators. 

Having stated the lattice theoretic properties of reactive processes, we move onto the relational
operators. Intuitively, $\healthy{R1}$ and $\healthy{R2}_c$ together ensure that the reactive behaviour of a process
contributes an extension $t$ to the trace.

\begin{theorem}[$\healthy{R1}$-$\healthy{R2}_c$ trace contribution] \label{thm:trcontr}
$$\healthy{R1}(\healthy{R2}_c(P)) = (\exists t @ P[\tempty,t/tr,tr'] \land tr' = tr \tcat t)$$
\end{theorem}

\noindent This shows that for any $\healthy{R1}$-$\healthy{R2}_c$ process there exists a trace extension $t$
recording its behaviour, and that $tr'$ is the prior history appended with this extension. Aside from illustrating
$\healthy{R1}$ and $\healthy{R2}_c$, this allows us to restate a process containing $tr$ and $tr'$ to one with only the
extension logical variable $t$, which provides a more natural entry point for reasoning about the trace
contribution of a process. In particular, we can prove a related law about sequential composition of reactive processes.

\begin{theorem}[$\healthy{R1}$-$\healthy{R2}_c$ sequential] If $P$ and $Q$ are \healthy{R1}-$\healthy{R2}_c$ healthy, then
  \begin{align*}
    P \relsemi Q = \exists t_1, t_2 @ ((& P[\tempty,t_1/tr,tr'] \relsemi \\ 
                                                      & Q[\tempty,t_2/tr,tr']) \land \\
                                                      & tr' = tr \tcat t_1 \tcat t_2)
  \end{align*}
\end{theorem}

\begin{proof}
  By Theorem~\ref{thm:trcontr} and relational calculus.
\end{proof}

\noindent This theorem shows that two sequentially composed processes have their own unique contribution to the trace
without sharing or interference. When applied in the context of a timed trace, for example, it allows us to subdivide
the trajectory into segments, which we can reason about separately. This theorem allows us to demonstrate closure of
$\healthy{R1}$-$\healthy{R2}_c$ predicates under sequential composition.

\begin{theorem}[$\healthy{R1}$-$\healthy{R2}_c$ sequential closure] \label{thm:r1r2seq} If $P$ and $Q$ are both
  \healthy{R1} and $\healthy{R2}_c$ healthy then $$\healthy{R1}(\healthy{R2}_c(P \relsemi Q)) = P \relsemi Q$$
\end{theorem}

\noindent Closure of $\healthy{R3}$ has previously been shown~\cite{Cavalcanti&06} and we have mechanised this proof.
This allows us to prove the following theorem.

\begin{theorem}[$\healthy{R}$ sequential closure] If $P$ and $Q$ are both \healthy{R} healthy then $P \relsemi Q$ is
  \healthy{R} healthy.
\end{theorem}

\noindent We have now shown that reactive processes are closed under the lattice and relational operators, and can use
these results to demonstrate the algebraic nature of the theory, by showing that reactive processes form a weak unital
quantale.

\begin{theorem} $\healthy{R}$ predicates form a weak unital quantale. Provided $A \subseteq \carrier{\healthy{R}}$ and $A \neq \emptyset$ the following
  laws hold:
  \begin{align*}
    P \relsemi (\thinf{R} \, A) &= (\thinf{R} Q\in A @ P \relsemi Q) \tag{Q1} \label{law:Q1} \\
    (\thinf{R} \, A) \relsemi Q &= (\thinf{R} P\in A @ P \relsemi Q) \tag{Q2} \label{law:Q2} \\
    P \relsemi \II &=  \II \relsemi P = P \tag{Q3} \label{law:Q3}
  \end{align*}
\end{theorem}

\begin{proof}
  Since $\thinf{R} \, A = \healthy{R}(\bigsqcap A)$ and sequential composition left and right distributes through
  $\bigsqcap$ it suffices, to show that $\healthy{R}$ is continuous: it distributes through non-empty infima.
\end{proof}

\noindent \ref{law:Q1} and \ref{law:Q2} are the quantale laws, which state that sequential composition distributes
through infima. The requirement of non-emptiness is why the quantale is called ``weak''. Finally, \ref{law:Q3} makes the weak
quantale unital. Unital quantales are an important algebraic structure that give rise to Kleene
algebras~\cite{Armstrong2015}. They augment a complete lattice with the laws above, the combination of which provides a
minimal algebraic foundation for substantiating the point-free laws of sequential programming~\cite{Armstrong2015}.

Our final result is closure under parallel composition. The UTP provides an operator called
\emph{parallel-by-merge}~\cite{Hoare&98}, $P \parallel_M Q$, whereby the composition of processes $P$ and $Q$ separates
their states, calculates their independent concurrent behaviours, and then merges the results. The operator is
parametric over merge predicate $M$ that specifies how synchronisation is performed. Different programming language
semantics require formation of a bespoke merge predicate depending on their concurrency scheme. We give a slightly
simplified version of the UTP definition, which is nevertheless equivalent.

\begin{definition}[Parallel-by-merge]
  $$P \parallel_M Q ~\defs~ (\psep{P}{0} \land \psep{Q}{1} \land v' = v) \relsemi M$$
\end{definition}

\noindent Operator $\psep{P}{n}$ augments the after variables of $P$ with an index; for example:
$$\psep{x' = 7 \cdot y}{0} = (0.x' = 7 \cdot y)$$ The three conjuncts rename the after variables of $P$ and $Q$ to
ensure no clashes, and copy all before variables ($v$) to after variables, respectively. Thus, $M$ has access to the
state of each variable before execution ($v$), and from the respective composed processes ($0.v$ and $1.v$). Merge
predicate $M$ can then invoke $tr' = f(0.tr, 1.tr)$ with a suitable trace merge function $f$, such as interleaving.

The healthiness conditions $\healthy{R1}$ and $\healthy{R3}$ can be directly applied to $M$, modulo some differences in
alphabet. $\healthy{R2}_c$ requires adaptation as it is possible to access the trace history through the two indexed
traces, $0.tr$ and $1.tr$, in addition to $tr$. It is, therefore, necessary to delete the history from the two in the
revised healthiness condition $\healthy{R2}_m$ below.

\begin{definition}[$\healthy{R2}_c$ for merge predicates]
  \begin{align*}
    \healthy{R2}_m(M) &\defs (P[\tempty,tr'\!-\!tr,0.tr\!-\!tr,1.tr\!-\!tr \\
                      & \qquad~~ /tr,tr',0.tr,1.tr]) \infixIf tr \le tr' \infixElse P
  \end{align*}
\end{definition}

\noindent $\healthy{R2}_m$ has the same form as $\healthy{R2}_c$ except that it deletes the history of three extant
traces, $tr'$, $0.tr$, and $1.tr$. From $M$'s perspective, $0.tr$ and $1.tr$ contain the trace the parallel processes
have executed. Thus we need to delete the history, through substitution, from these as well so that they contain only
the contributions of their respective processes. This allows us to show that the overall composition is
$\healthy{R2}_c$. We define a condition for merge predicates -- $\healthy{R}_m \defs \healthy{R1} \circ \healthy{R2}_m
\circ \healthy{R3}$ -- and prove the following final theorem.

\begin{theorem}
  $P \parallel_M Q$ is $\healthy{R}$ healthy provided that $P, Q$ are $\healthy{R}$ healthy, and $M$ is
  $\healthy{R}_m$ healthy.
\end{theorem}

\noindent Thus our generalised theory of reactive processes is conservative and unifies the
denotational semantics of concurrent programming.

\section{Conclusion}
Traces are ubiquitous in modelling of program history. Here, we have shown how a generalised foundation for their
semantics can be given in terms of a trace algebra, and presented some important models, notably piecewise-continuous
functions. Finally, we have applied it to reconstruct Hoare and He's model of reactive processes, with some important
additions of our own, including revision of $\healthy{R2}$, additional theorems about reactive relations, and lifting of
the healthiness conditions to parallel composition. All of the theorems described herein have been mechanised in
Isabelle/UTP~\cite{Foster16a}.

In the future we will apply this theory of reactive processes to give a new model to the UTP hybrid relational
calculus~\cite{Foster16b} that we have previously created to give denotational semantics to Modelica and
Simulink. Moreover, inspired by \cite{Oliveira&09}, we will use our theory to describe generalised reactive designs, a
UTP theory that justifies combined use of concurrent and assertional reasoning. This will enable the construction of
verification tools on top of our Isabelle/HOL embedding for concurrent and hybrid programming languages. 

We also aim to explore weakenings of the trace algebra and healthiness conditions to support larger classes of reactive
process semantics. For example, weakening of the trace cancellation laws could enable representation of infinite traces
in order to support reactive processes with unbounded nondeterminism. Moreover, $\healthy{R2}_c$ currently prevents a
process from depending on an absolute start time with respect to a global clock. In the future this could be relaxed,
either at the model or theory level, to support time variant real-time and hybrid processes.

\section*{Acknowledgements}

\noindent This work is funded by EU H2020 project ``INTO-CPS''\footnote{\url{http://into-cps.au.dk}}, grant agreement
644047. We would also like to thank Dr. Jeremy Jacobs, and also our anonymous reviewers, for their helpful feedback.

\bibliography{grproc}

\end{document}